\documentclass[12pt,draftclsnofoot,onecolumn]{IEEEtran}

\newcommand{\diag}{\mathop{\rm diag}}

\makeatletter
\newcommand{\vast}{\bBigg@{4}}
\newcommand{\vastt}{\bBigg@{6}}
\makeatother
\usepackage[final]{graphicx}
\usepackage{color}
\usepackage{multirow}
\usepackage{graphicx}
\usepackage{epstopdf}
\usepackage[cmex10]{amsmath}
\usepackage{amssymb}

\usepackage{multirow}
\usepackage{algpseudocode}
\usepackage{amsthm}
\usepackage{cite}
\usepackage{algorithm, algcompatible}
\theoremstyle{definition}
\theoremstyle{remark}
\newtheorem{rem}{Remark}
\newtheorem{lem}{Lemma}

\newtheorem{cor}{Corollary}
\newtheorem{prop}{Proposition}

\makeatletter
\g@addto@macro\th@remark{\thm@headpunct{\normalfont:}}
\makeatother
\makeatletter
\newcommand{\distas}[1]{\mathbin{\overset{#1}{\kern\z@\sim}}}%
\newsavebox{\mybox}\newsavebox{\mysim}
\newcommand{\distras}[1]{%
  \savebox{\mybox}{\hbox{\kern3pt$\scriptstyle#1$\kern3pt}}%
  \savebox{\mysim}{\hbox{$\sim$}}%
  \mathbin{\overset{#1}{\kern\z@\resizebox{\wd\mybox}{\ht\mysim}{$\sim$}}}%
	}

\allowdisplaybreaks

\algnewcommand\INPUT{\item[\textbf{Input:}]}
\algnewcommand\OUTPUT{\item[\textbf{Output:}]}

\begin{document}

\title{All Cognitive MIMO: A New Multiuser Detection Approach with Different Priorities}

\author{Nikolaos~I.~Miridakis, Theodoros~A.~Tsiftsis,~\IEEEmembership{Senior Member,~IEEE},\\ Dimitrios~D.~Vergados,~\IEEEmembership{Senior Member,~IEEE}, and Angelos~Michalas,~\IEEEmembership{Member,~IEEE}
\thanks{N. I. Miridakis is with the Department of Computer Engineering, Piraeus University of Applied Sciences, 12244 Aegaleo, Greece (e-mail: nikozm@unipi.gr).}
\thanks{T. A. Tsiftsis is with the School of Engineering, Nazarbayev University, 010000 Astana, Kazakhstan (e-mails: tsiftsis@teiste.gr, theodoros.tsiftsis@nu.edu.kz).}
\thanks{D. D. Vergados is with the Department of Informatics, University of Piraeus, 18534, Piraeus, Greece (email: vergados@unipi.gr).}
\thanks{A. Michalas is with the Department of Informatics Engineering, Technological Education Institute of Western Macedonia, 52100, Kastoria, Greece (e-mail: amichalas@kastoria.teiwm.gr).}
}

\markboth{}%
{}
\maketitle

\begin{abstract}
A new detection scheme for multiuser multiple-input multiple-output (MIMO) systems is analytically presented. In particular, the transmitting users are being categorized in two distinct priority service groups, while they communicate directly with a multi-antenna receiver. The linear zero-forcing scheme is applied in two consecutive detection stages upon the signal reception. In the first stage, the signals of one service group are detected, followed by the second stage including the corresponding detection of the remaining signals. An appropriate switching scheme based on specific transmission quality requirements is utilized prior to the detection so as to allocate the signals of a given service group to the suitable detection stage. The objective is the enhancement of the reception quality for both service groups. The proposed approach can be implemented directly in cognitive radio communication assigning the secondary users to the appropriate service group. The exact outage probability of the considered system is derived in closed form. The special case of massive MIMO is further studied yielding some useful engineering outcomes; the effective channel coherence time and a certain optimality condition defining both the transmission quality and effective number of independent transmissions.
\end{abstract}


\IEEEpeerreviewmaketitle

\section{Introduction}
\IEEEPARstart{T}{he} scarcity of wireless bandwidth prompts the need for spectrally efficient methods. Multiple-input multiple-output (MIMO) communication is widely acknowledged as a key technology for achieving high data rates in bandwidth-constrained wireless systems \cite{j:MietznerSchober2009}. Recently, there has been a considerable research effort in exploiting the space dimension through spatial multiplexing methods so as to enhance the multiuser diversity; especially in dense communication systems \cite{j:NgoLarsson2013}. Specifically, spatial multiplexing can be used to transmit independent multiple data streams that can be separated using adequate signal processing at the receiver. Nowadays, recent advances of MIMO deployments allow an increasing scale on the order of antenna arrays, up to several hundreds of antennas, yielding to the so-called very-large scale (or massive) MIMO \cite{j:massive,j:bjornson}. Moving towards to this trend, the multiuser access for spatial multiplexed transmission systems can be further increased, reflecting on the overall enhancement of spectral and data rate efficiency. 

On another front, cognitive radio (CR) has emerged as one of the most promising technologies to resolve the issue of spectrum scarcity \cite{j:Haykin,j:LiangChen2011}. The key enabling technology of CR relies on the provision of capability to share the spectrum in an opportunistic manner, such that the secondary transmissions do not cause any harmful interference to the primary communication. Due to the complementary benefits of MIMO and CR, hybrid CR MIMO systems are of prime research interest \cite{j:NaeemRehmani2017}.

To date, three cornerstone paradigms explore the coexistence between primary and secondary systems; the interweave, underlay, and overlay CR schemes \cite{j:GoldsmithJafar2009}. These schemes provide different modes of operation for the secondary nodes, according to the application-dependent scenario. A fundamental feature of the interweave and underlay schemes is that the primary system is fully unaware of the existence of any underlying CR activity. On the other hand, overlay scheme admits a certain level of cooperation between primary and secondary systems (e.g., see \cite{j:SimeoneStanojev2008,j:SuMatyjas2012,j:MannaLouie2011,j:SongHasna2013,j:ZhengSong2013,j:ZhengKrikidis2013} and references therein). In such cooperation scenarios, secondary users serve as relays, by forwarding the primary users' data streams. In return, a certain resource fraction (e.g., time and/or frequency) is reserved for CR-only use. The shift from a single-antenna \cite{j:SimeoneStanojev2008,j:SuMatyjas2012} to multiple-antenna \cite{j:MannaLouie2011,j:SongHasna2013,j:ZhengSong2013,j:ZhengKrikidis2013} regime has shown a tremendous increase in achievable data rate regions, while it reduces the resource fraction allocated only for the CR system.

So far, these efforts have been limited to only having the CR nodes assist the primary nodes' traffic via relaying. There is no consideration of the converse or a broader vision of a \emph{policy-based} cooperation between the two communication systems. In fact, the latter limitation is mainly due to the current Federal Communications Commission (FCC) requirements. Still, the recent works in \cite{c:YuanShi2013,j:YuanTian2016} studied the performance of a full cooperative communication between the two systems, on a network-level scenario, where both primary and secondary nodes serve as relays to forward data streams to one another. A fundamental question arises when adopting this case: What would be the benefit from adopting such an approach? Interestingly, it was explicitly demonstrated that the throughput regions of both systems are greatly enhanced when applying full cooperation \cite{j:YuanTian2016}. Yet, single-antennas were assumed in \cite{c:YuanShi2013,j:YuanTian2016}. The scenario of a fully cooperative communication between primary and secondary MIMO systems has not been studied in the open technical literature. 

Capitalizing on the aforementioned observations, we introduce a novel wireless communication approach, which is termed as `\emph{all cognitive MIMO system}'. The envisioned system consists of multiple primary and secondary nodes, randomly placed in a given geographical area. These nodes simultaneously transmit their data streams to a common receiver, which is equipped with multiple antennas; thereby, forming a multiuser MIMO system. This system is different from the conventional multiuser MIMO, such that it consists of two distinct priority classes. Particularly, primary and secondary nodes are grouped together in distinct service groups, Service~1 and Service~2, correspondingly. The rather practical (due to the efficient tradeoff between performance and complexity) linear processing is achieved at the receiver \cite{j:MiridakisVergados2013}. In this work, the linear zero-forcing (ZF) scheme is adopted for the multiuser signal detection. Each service group is described by different priorities and/or transmission quality requirements. Clearly, Service~1 gets the highest priority, representing the licensed primary network nodes. Doing so, the receiver guarantees the transmission quality of these nodes. Only in the case when the latter requirement is satisfied, then the receiver strives for the enhancement of the transmission quality for Service~2.

Focusing on spatial multiplexed transmission, the receiver simultaneously captures the multiple signals from the two service groups. To improve the reception quality, the ZF filter is utilized in two consecutive stages. Therefore, the streams that are being detected/decoded in the second stage experience less inter-stream interference since the contribution of signals detected in the first stage have already been removed. As a countermeasure, the streams of the second stage experience a higher system latency due to the two-stage detection process. Further, a switching scheme is introduced, aiming to ensure, and more than this, to enhance the aforementioned transmission quality for each service group. The role of the switching scheme is to efficiently allocate/reallocate the streams of Service~1 and Service~2 in the first and second detection stage, respectively, or vice versa.  

At this point, it is noteworthy that in future wireless networks, machine-type communications such as the Internet of Things, Internet of Everything, Smart X, etc. are expected to play an important role. In such applications, connectivity rather than high throughput seems to be more important. As an illustrative example, transmitting devices of this type can be considered as secondary nodes, assigned in Service~2 of the considered approach. Even in dense primary networks, there may be certain time instances where a fraction of system resources could be available (occasionally). Hence, the proposed scheme can effectively be utilized to serve both systems simultaneously, as described above.

Overall, the main benefits of the proposed approach are summarized as:
\begin{itemize}
	\item A single receiver is used for both primary and secondary communication. Thus, the ever increasing economical and environmental (e.g., carbon footprint) costs, associated with the operating expenditure of communication networks, is reduced.
	\item The considered full cooperation between the primary and secondary system can be visualized as follows: In the downlink of an infrastructure-based system, secondary nodes may serve as relays to assist the primary communication (e.g., as in \cite{j:SimeoneStanojev2008,j:SuMatyjas2012,j:MannaLouie2011,j:SongHasna2013,j:ZhengSong2013,j:ZhengKrikidis2013}). In return, the secondary nodes may access the network via the common receiver, in the uplink direction, by applying the proposed approach (i.e., this case is analytically described in subsection \ref{System Model}.A).
	\item Both the primary and secondary nodes are able to use a full transmission power profile without affecting each others.
	\item The introduced switching scheme at the detector enhances the performance of both primary and secondary communication. 
	\item The adoption of linear detection with a moderate computational burden facilitates the applicability of the considered scheme in realistic conditions.
\end{itemize}

In what follows, the performance of the introduced `all cognitive MIMO' system is analytically studied, whereas some useful engineering outcomes are extracted. The main contributions of this work can be summarized as follows:
\begin{itemize}
	\item Rayleigh channel fading conditions are assumed. Moreover, the channel gains are modeled as independent and non-necessarily identically distributed (i.n.n.i.d.) RVs. This reflects to the case of arbitrary link distances between the involved nodes and the receiver, an appropriate condition for practical applications. The extreme scenarios of independent and non-identically distributed (i.n.i.d.) and independent and identically distributed (i.i.d.) received channel gains are also considered as special cases. Moreover, the joint imperfect channel state information (CSI) is considered along with the so-called channel aging effect (i.e., due to the movement of transmitting nodes and/or fast channel fading variations). 
	\item The second detection stage, except from the imperfect CSI, may also suffer from residual noise caused by imperfect signal decoding (and thus signal removal) during the first detection stage. Such a scenario is explicitly defined in the included analysis.
	\item The exact outage performance of the considered system is derived in a closed-form expression for the general channel fading scenario including arbitrary antenna arrays.
	\item The scenario of massive MIMO is further studied, standing for an illustrative example of prime interest. In such an `all cognitive massive MIMO' concept, the channel coherence time is analytically derived. In addition, the admissible number of secondary nodes is also defined by deriving a simple yet necessary and sufficient optimality condition.   
\end{itemize}

The rest of this paper is organized as follows: In Section \ref{System Model}, the considered system model is explicitly defined along with the proposed mode of operation. Section \ref{SNR Statistics} provides the most important statistics of the signal-to-noise ratio (SNR) for the received streams. Leveraging the latter results, Section \ref{Performance Metrics} presents some key performance metrics, namely, the outage probability of the system for the general scenario and the effective channel coherence time along with the optimal number of admissible transmitting secondary nodes for the massive MIMO scenario. Selected numerical results and useful discussions are presented in Section \ref{Numerical Results}, while some concluding remarks are given in Section \ref{Conclusion}.

\emph{Notation}: Vectors and matrices are represented by lowercase bold typeface and uppercase bold typeface letters, respectively. Also, $\mathbf{X}^{\dagger}$ is the Moore-Penrose pseudo-inverse of $\mathbf{X}$, $[\mathbf{X}]_{i}$ represents the $i^{\rm th}$ row of $\mathbf{X}$, and $\mathbf{x}_{i}$ denotes the $i^{\rm th}$ coefficient of $\mathbf{x}$. A diagonal matrix with entries $x_{1},\cdots,x_{n}$ is defined as $\diag\{x_{i}\}^{n}_{i=1}$. The superscripts $(\cdot)^{T}$ and $(\cdot)^{\mathcal{H}}$ denote transposition and Hermitian transposition, respectively, while $\|\cdot\|$ corresponds to the vector Euclidean norm. In addition, $\mathbf{I}_{v}$ stands for the $v\times v$ identity matrix, $\mathbb{E}[\cdot]$ is the expectation operator, $\overset{\text{d}}=$ represents equality in probability distributions and $\text{Pr}[\cdot]$ returns probability. Also, $f_{X}(\cdot)$, $F_{X}(\cdot)$, and $\bar{F}_{X}(\cdot)$ represent probability density function (PDF), cumulative distribution function (CDF), and complementary CDF (CCDF) of the random variable (RV) $X$, respectively. A complex-valued Gaussian RV with mean $\mu$ and variance $\sigma^{2}$ is denoted as $\mathcal{CN}(\mu,\sigma^{2})$. Furthermore, $\Gamma(\cdot)$ denotes the Gamma function \cite[Eq. (8.310.1)]{tables}, $\Gamma(\cdot.,\cdot)$ is the upper incomplete Gamma function \cite[Eq. (8.350.2)]{tables}, while $J_{0}(\cdot)$ represents the zeroth-order Bessel function of the first kind \cite[Eq. (8.441.1)]{tables}. Finally, $\left\lfloor \cdot \right\rfloor$ is the floor function, while $\mathcal{O}(\cdot)$ is the Landau symbol (i.e., $f(x)=\mathcal{O}(g(x))$, when $|f(x)|\leq v |g(x)| \ \forall x\geq x_{0}$, $\left\{v,x_{0}\right\} \in \mathbb{R}$).

\section{System Model}
\label{System Model}
Consider a multiuser communication system, where the transmitted streams can be groupped together in two distinct service groups. Each service groups corresponds to a different transmission quality level and/or priority. Let $M_{1}$ single-antenna nodes comprise the first service group, entitled as Service 1, and $M_{2}$ single-antenna nodes comprise the second group, entitled as Service 2.\footnote{It follows from the subsequent analysis that the considered system is equivalent to the case when a single transmitter for each service group is used, equipped with $M_{1}$ and $M_{2}$ antennas, respectively.} The system operates under the presence of a (common) receiver equipped with $N\geq M$ antennas, where $M\triangleq M_{1}+M_{2}$. Moreover, i.n.n.i.d. Rayleigh flat fading channels are assumed, reflecting arbitrary distances among the involved nodes with respect to the receiver, an appropriate condition for practical applications.

The spatial multiplexing mode of operation is implemented, where $M$ independent data streams are simultaneously transmitted by the corresponding transmitting nodes. The suboptimal yet quite efficient ZF detection scheme is adopted at the receiver.

The received signal at the $n^{\rm th}$ sample time-instance reads as
\begin{align}
\mathbf{y}[n]=\mathbf{H}[n] \mathbf{P}^{\frac{1}{2}}[n] \mathbf{s}[n]+\mathbf{w}[n],
\label{eq1}
\end{align} 
where $\mathbf{y}[n] \in \mathbb{C}^{N\times 1}$, $\mathbf{H}[n] \in \mathbb{C}^{N\times M}$, $\mathbf{P}[n]\in \mathbb{R}^{M\times M}$, $\mathbf{s}[n] \in \mathbb{C}^{M\times 1}$ and $\mathbf{w}[n] \in \mathbb{C}^{N\times 1}$ denote the received signal, the channel matrix, the received signal power, the transmitted signal and the additive white Gaussian noise (AWGN), respectively. It holds that $\mathbf{w}[n]\overset{\text{d}}=\mathcal{CN}(\mathbf{0},N_{0}\mathbf{I}_{N})$ with $N_{0}$ denoting the AWGN variance and $\mathbf{s}[n]=[\mathbf{s}^{T}_{1}[n]\ \ \mathbf{s}^{T}_{2}[n]]^{T}$ with $\mathbb{E}[\mathbf{s}[n]\mathbf{s}^{\mathcal{H}}[n]]=\mathbf{I}_{M}$, while $\mathbf{s}_{1}[n] \in \mathbb{C}^{M_{1}\times 1}$ and $\mathbf{s}_{2}[n] \in \mathbb{C}^{M_{2}\times 1}$ are i.i.d. zero-mean unit-variance RVs. In addition, $\mathbf{H}[n]=[\mathbf{H}_{1}[n]\ \ \mathbf{H}_{2}[n]]$ with $\mathbf{H}_{1}[n] \in \mathbb{C}^{N\times M_{1}}$ and $\mathbf{H}_{2}[n] \in \mathbb{C}^{N\times M_{2}}$. Also, $\mathbf{h}_{i}[n]\overset{\text{d}}=\mathcal{CN}(\mathbf{0},\mathbf{I}_{N})$ is a column vector of $\mathbf{H}[n]$. Finally, $\mathbf{P}[n]\triangleq \diag\{p_{i}\}^{M}_{i=1}$ with $p_{i}\triangleq p_{T}d_{i}^{-\omega_{i}}$, where $p_{T}$, $d_{i}$, and $\omega_{i}\in [2,6]$ correspond to the transmit power,\footnote{Without loss of generality, an equal transmit power profile is considered for all the involved transmitting nodes.} normalized distance (with a reference distance equal to $1$km) from the receiver, and path-loss exponent of the $i^{\rm th}$ transmitter, respectively.

\subsection{Mode of Operation}
The entire system communication is accomplished in consecutive frames. Each frame is divided in two parts/phases; the training phase and the data phase. In the former phase, each node transmits a unique pilot signature to the receiver to provide CSI. Afterwards, the latter phase occurs, where the $M$ transmitters simultaneously send their data streams at the remaining frame duration. The latter duration is directly related to the channel coherence time, determining the efficiency of the prior channel estimation.

In this paper, Service 1 denotes the primary service with the highest priority between the two services. More specifically, the system should preserve a predetermined transmission quality for the $M_{1}$ received streams, regardless of the presence or not of the other $M_{2}$ streams for Service 2. The transmission quality is measured with the aid of a certain SNR threshold, defined as $\gamma_{\rm T}\triangleq 2^{\mathcal{R}}-1$, reflecting a minimum data rate requirement $\mathcal{R}$ (in bps/Hz) for Service~1. Note that the transmitting nodes of each service group are fully unaware of the existence of the other service group. Only the receiver has a full overview of the system status, which guarantees the transmission quality of Service 1, while it strives for the enhancement of the transmission quality of Service 2. \emph{Due to the full awareness of the receiver, the considered MIMO system is termed as `all cognitive' providing access to both primary and secondary nodes, yet with different priorities}.  

Upon the CSI acquisition, the system enters the data phase. Doing so, the receiver captures the $M$ independent streams and performs signal detection so as to decode the corresponding data, until the next training phase. The ZF equalizer is utilized for the multiple signal detection. In the case when the transmission quality of \emph{all} the $M_{1}$ streams of Service 1 is sufficient, they are detected/decoded straightforwardly. Then, the remaining signal (which consists of the $M_{2}$ received streams of Service 2) enters a second-stage ZF detection process so as to obtain the corresponding transmitted data. The proposed approach provides the following features:
\begin{enumerate}
	\item The streams of Service 1 suffer from an increased inter-stream interfering power (i.e., $M$-fold) in comparison to the streams of Service 2, which experience only intra-stream interference (i.e., $M-M_{1}=M_{2}$-fold), due to the removal of $M_{1}$ symbols from the remaining signal, performed in the first detection stage.
	\item On the other hand, the processing delay and system latency regarding the $M_{2}$ streams of Service 2 is increased as compared to Service 1 because of the two-stage detection process. 
\end{enumerate}

Nevertheless, in the case when the transmission quality of Service 1 is not satisfied, then a switching process between the two service groups occurs at the receiver. In particular, the $M_{2}$ streams of Service 2 are detected/decoded first and then the remaining $M_{1}$ streams of Service 1 enter the second-stage detection (c.f. Fig.~\ref{fig1}). The main objective is the enhancement of the transmission quality of Service 1 (i.e., see feature 1 above) at the cost of a slightly increased system latency. In fact, the introduced processing delay is directly related to the computational complexity upon the signal detection. The complexity of the most efficient ZF filter (i.e., the pseudo-inverse operation) follows $\mathcal{O}(N^{q})$ with $2<q<3$ \cite{j:Coppersmith1990}. Keeping in mind that the signal processing is performed at the receiver (which is usually a base-station with a fixed power supply and advanced computational capabilities), the latter computational burden is not an issue for a typical low-to-moderate antenna range (i.e., $N\leq 8$). Nonetheless, it is obvious that the complexity emphatically increases for quite a high antenna regime. In this case, it is preferable to keep the streams of Service 1 in the first detection stage and optimally determine the admissible $M_{2}$ streams for Service 2 at the second detection stage (this case is thoroughly studied in subsection \ref{Performance Metrics}.B). The basic lines of reasoning of the proposed detection scheme are formalized in Algorithm~1.

\begin{algorithm}[t]
	\caption{Steps of the Proposed Detector}
	\begin{algorithmic}[1]
		\INPUT{$M_{1}$, $M_{2}$, $N$, and the corresponding channel statistics}
		 \IF{the extra complexity $\propto \mathcal{O}(N^{q})$ for Service~1 is acceptable}
			\STATE{Proceed to a switching check;}
				\IF{$\{{\rm SNR}_{i}\}^{M_{1}}_{i=1}>\gamma_{\rm T}$}
					\STATE Service~1 $\rightarrow$ $1^{\rm st}$ detection stage;
					\STATE Service~2 $\rightarrow$ $2^{\rm nd}$ detection stage;
					\STATE End of Algorithm;
				\ELSE
					\STATE Service~1 $\rightarrow$ $2^{\rm nd}$ detection stage;
					\STATE Service~2 $\rightarrow$ $1^{\rm st}$ detection stage;
					\STATE End of Algorithm;
				\ENDIF
		 \ELSE \:(e.g., in massive MIMO)
		   \STATE Service~1 $\rightarrow$ $1^{\rm st}$ detection stage;
			 \STATE Determine the optimal  value of $M_{2}$ via Algorithm~2 (see subsection \ref{Performance Metrics}.B);
			 \STATE End of Algorithm;
		\ENDIF 
	\end{algorithmic}
\end{algorithm} 

\begin{figure}[!t]
\centering
\includegraphics[trim=0.0cm 0.0cm 0.0cm 0.0cm, clip=true,totalheight=0.198\textheight]{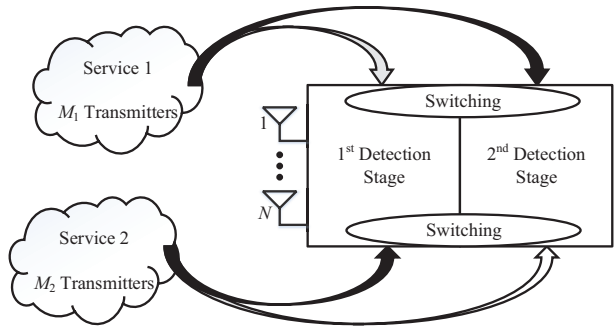}
\caption{Basic illustration of the proposed detection approach.}
\label{fig1}
\end{figure}

Overall, the main benefit of the aforementioned mode of operation is the enhancement of multiuser diversity, while maintaining a predefined transmission quality at the same time. Subsequently, the training and data phases of the proposed approach are explicitly defined and analyzed.   

\subsection{Training Phase: Channel Estimation}
During the training phase, $M$ orthogonal pilot sequences (i.e., unique spatial signal signatures) of length $M$ symbols are assigned to the involved nodes. Then, the received pilot signal can be expressed as
\begin{align}
\mathbf{Y}_{\text{tr}}[n]=\mathbf{H}_{\text{tr}}[n]\mathbf{P}^{\frac{1}{2}}_{\text{tr}}[n] \mathbf{\Psi}+\mathbf{W}_{\text{tr}}[n],
\label{trainingsignal}
\end{align}  
where $\mathbf{Y}_{\text{tr}}[n] \in \mathbb{C}^{N\times M}$, $\mathbf{H}_{\text{tr}}[n] \in \mathbb{C}^{N\times M}$, $\mathbf{P}_{\text{tr}}[n]\in \mathbb{R}^{M\times M}$, $\mathbf{\Psi} \in \mathbb{C}^{M\times M}$ and $\mathbf{W}_{\text{tr}}[n] \in \mathbb{C}^{N\times M}$ denote the received signal, the channel matrix, the received signal power, the transmitted pilot signals and AWGN, respectively, all representing the training phase. Also, the pilot signals are normalized satisfying $\mathbb{E}[\mathbf{\Psi}\mathbf{\Psi}^{\mathcal{H}}]=\mathbf{I}_{M}$.

Using MMSE channel estimation, the channel estimate and the channel estimation error remain uncorrelated (i.e., due to the orthogonality principle \cite{b:stevenkayesttheory}). In particular, we have that
\begin{align}
\hat{\mathbf{h}}_{i}[n]\hat{\mathbf{p}}^{\frac{1}{2}}_{i}[n]=\mathbf{h}_{i}[n]\mathbf{p}^{\frac{1}{2}}_{i}[n] +\tilde{\mathbf{h}}_{i}[n]\tilde{\mathbf{p}}^{\frac{1}{2}}_{i}[n],\ \ 1\leq i\leq M,
\label{channel decomposition}
\end{align}
where $\mathbf{h}_{i}\mathbf{p}^{\frac{1}{2}}_{i}\overset{\text{d}}=\mathcal{CN}(\mathbf{0},(p_{i}-\hat{p}_{i})\mathbf{I}_{N})$ is the true channel fading of the $i^{\rm th}$ transmitter and $\tilde{\mathbf{h}}_{i}\tilde{\mathbf{p}}^{\frac{1}{2}}_{i}\overset{\text{d}}=\mathcal{CN}(\mathbf{0},\hat{p}_{i}\mathbf{I}_{N})$ denotes its corresponding estimation error with \cite[Eq. (12)]{j:Truong_massive_Mimo}
\begin{align}
\hat{p}_{i}\triangleq \frac{p^{2}_{i}}{\left(p_{i}+\frac{1}{Mp_{T}}\right)}.
\label{esterror}
\end{align}
Except from the channel estimation errors, the channel aging effect occurs in several practical network setups. This is mainly because of the rapid channel variations during consecutive sample time-instances, due to, e.g., user mobility and/or severe fast fading conditions. The popular autoregressive (Jakes) model of a certain order \cite{jakes}, based on Gauss-Markov block fading channel, can accurately capture the latter effect. More specifically, it holds that
\begin{align}
\nonumber
\hat{\mathbf{h}}_{i}[n]\hat{\mathbf{p}}^{\frac{1}{2}}_{i}[n]=&\alpha^{M}\hat{\mathbf{h}}_{i}[n-M]\hat{\mathbf{p}}^{\frac{1}{2}}_{i}[n-M]\\
&+\sum^{M-1}_{m=0}\alpha^{m}\mathbf{e}_{i}[n-m],
\label{jakesorderm}
\end{align}
where $\alpha\triangleq J_{0}(2\pi f_{D}T_{s})$ with $f_{D}$ and $T_{s}$ denoting the maximum Doppler shift and the symbol sampling period, respectively. Moreover, $\mathbf{e}'_{i}\triangleq \sum^{M-1}_{m=0}\alpha^{m}\mathbf{e}_{i}[n-m]$ stands for the stationary Gaussian channel error vector due to the time variation of the channel, which is uncorrelated with $\mathbf{h}_{i}[n-M]\mathbf{p}^{\frac{1}{2}}_{i}[n-M]$, while 
\begin{align}
\mathbf{e}'_{i}\overset{\text{d}}=\mathcal{CN}(\mathbf{0},(1-\alpha^{2 M})p_{i}\mathbf{I}_{N}).
\label{errorderm}
\end{align}
For the sake of mathematical simplicity and without loss of generality, we assume that the channel remains unchanged over the time period of training phase, while it may change during the subsequent data transmission phase. Thus, adopting the autoregressive model of order one, we get 
\begin{align}
\hat{\mathbf{h}}_{i}[n]\hat{\mathbf{p}}^{\frac{1}{2}}_{i}[n]=\alpha \hat{\mathbf{h}}_{i}[n-1]\hat{\mathbf{p}}^{\frac{1}{2}}_{i}[n-1] +\mathbf{e}_{i}[n].
\label{channel aging1}
\end{align} 
It readily follows that\footnote{In what follows, the time-instance index $n$ is dropped for ease of presentation, since all the involved random vectors are mutually independent.} 
\begin{align}
\hat{\mathbf{H}}\hat{\mathbf{P}}^{\frac{1}{2}}=\mathbf{H}\mathbf{P}^{\frac{1}{2}}+\mathbf{E},
\end{align}
with
\begin{align}
\nonumber
\hat{\mathbf{H}}\hat{\mathbf{P}}^{\frac{1}{2}}&\overset{\text{d}}=\mathcal{CN}(\mathbf{0},a^{2}\hat{\mathbf{P}}),\\
\nonumber
\mathbf{H}\mathbf{P}^{\frac{1}{2}}&\overset{\text{d}}=\mathcal{CN}(\mathbf{0},\mathbf{P}),\\
\mathbf{E}&\overset{\text{d}}=\mathcal{CN}(\mathbf{0},\mathbf{P}-a^{2}\hat{\mathbf{P}}),
\label{channeljoint}
\end{align} 
where $\mathbf{E}\triangleq [\mathbf{e}_{1}\ \cdots \ \mathbf{e}_{M}]$ stands for the estimation error matrix. It should be noted that the latter model in (\ref{channeljoint}) combines the joint effect of channel aging and channel estimation error.

\subsection{Data Phase: Signal Detection}
Upon the training phase, the receiver captures the statistics of all the upcoming transmissions of each node, which occur in the subsequent data phase. Specifically, all the transmitters simultaneously send their streams at the receiver, whereas the latter node detects the total signal via a ZF equalizer, $(\hat{\mathbf{H}}\hat{\mathbf{P}}^{\frac{1}{2}})^{\dagger}$, such that 
\begin{align}
\nonumber
\mathbf{r}&=\left(\hat{\mathbf{H}}\hat{\mathbf{P}}^{\frac{1}{2}}\right)^{\dagger}\mathbf{y}\\
\nonumber
&=\left(\hat{\mathbf{H}}\hat{\mathbf{P}}^{\frac{1}{2}}\right)^{\dagger}\left(\hat{\mathbf{H}}\hat{\mathbf{P}}^{\frac{1}{2}}\mathbf{s}+\mathbf{E}\mathbf{s}+\mathbf{w}\right)\\
&=\mathbf{s}+\left(\hat{\mathbf{H}}\hat{\mathbf{P}}^{\frac{1}{2}}\right)^{\dagger}\mathbf{E}\mathbf{s}+\left(\hat{\mathbf{H}}\hat{\mathbf{P}}^{\frac{1}{2}}\right)^{\dagger}\mathbf{w}.
\label{postdetectsignal}
\end{align}
In the first detection stage, the $M_{1}$ received streams (which belong to Service~1) are simultaneously detected, decoded and then extracted from the total post-detected signal $\mathbf{r}$. 

Afterwards, the receiver applies a second-stage ZF filter so as to detect the remaining $M_{2}$ streams (which correspond to Service~2), yielding
\begin{align}
\nonumber
\mathbf{r}'&=\left(\hat{\mathbf{H}}_{2}\hat{\mathbf{P}}'^{\frac{1}{2}}\right)^{\dagger}\mathbf{y}'\\
\nonumber
&=\left(\hat{\mathbf{H}}_{2}\hat{\mathbf{P}}'^{\frac{1}{2}}\right)^{\dagger}\left(\hat{\mathbf{H}}_{2}\hat{\mathbf{P}}'^{\frac{1}{2}}\mathbf{s}_{2}+\mathbf{E}_{2}\mathbf{s}_{2}+\mathbf{w}'\right)\\
&=\mathbf{s}_{2}+\left(\hat{\mathbf{H}}_{2}\hat{\mathbf{P}}'^{\frac{1}{2}}\right)^{\dagger}\mathbf{E}_{2}\mathbf{s}_{2}+\left(\hat{\mathbf{H}}_{2}\hat{\mathbf{P}}'^{\frac{1}{2}}\right)^{\dagger}\mathbf{w}',
\label{postdetectsignal2}
\end{align} 
where $\mathbf{E}_{2}\triangleq [\mathbf{e}_{M_{1}+1}\ \cdots \ \mathbf{e}_{M}]$, $\hat{\mathbf{P}}'\triangleq \diag\{\hat{p}_{i}\}^{M}_{i=M_{1}+1}$, and $\mathbf{w}'\triangleq \sqrt{\beta} \mathbf{w}$ denotes the post-noise after the imperfect cancellation/removal of $M_{1}$ streams at the aforementioned first detection stage, while $\beta$ represents the noise uncertainty factor. In fact, $\mathbf{w}'$ is colored since it reflects the combined impact of AWGN and residual noise caused by the imperfect cancellation/removal of $M_{1}$ streams during the first detection stage. Conditioning on $\beta$, it holds that
\begin{align}
\mathbf{w}'\overset{\text{d}}=\mathcal{CN}(\mathbf{0},\hat{N}_{0}\mathbf{I}_{N}), 
\label{estnoise}
\end{align}
where $\hat{N}_{0}\triangleq \beta N_{0}$ represents the estimated noise variance in the second detection stage at the receiver for Service 2. Note that in ideal decoding conditions, $\hat{N}_{0}=N_{0}$ and $\beta=1$.

\subsection{Post-Effective SNR}
According to \eqref{postdetectsignal} and \eqref{postdetectsignal2}, the post-detected SNR for the streams of Service 1 and Service 2 are, respectively, expressed as
\begin{align}
\nonumber
&{\rm SNR}^{(1)}_{i}=\\
&\frac{1}{\left\|\left[\left(\hat{\mathbf{H}}\hat{\mathbf{P}}^{\frac{1}{2}}\right)^{\dagger}\right]_{i}\mathbf{E}\right\|^{2}+N_{0}\left\|\left[\left(\hat{\mathbf{H}}\hat{\mathbf{P}}^{\frac{1}{2}}\right)^{\dagger}\right]_{i}\right\|^{2}},\ i\in [1,M_{1}],
\label{snr1}
\end{align}
and 
\begin{align}
\nonumber
{\rm SNR}^{(2)}_{i}=&\frac{1}{\left\|\left[\left(\hat{\mathbf{H}}_{2}\hat{\mathbf{P}}'^{\frac{1}{2}}\right)^{\dagger}\right]_{i}\mathbf{E}_{2}\right\|^{2}+\left(\frac{\hat{N}_{0}}{\beta}\right)\left\|\left[\left(\hat{\mathbf{H}}_{2}\hat{\mathbf{P}}'^{\frac{1}{2}}\right)^{\dagger}\right]_{i}\right\|^{2}},\\
&\ \ \ \ \ \ \ \ \ \ \ \ \ \ \ \ \ \ \ \ \ \ \ \ \ \ \ \ \ \ \ \ \ \ \ \ \ \ i\in [1,M_{2}].
\label{snr2}
\end{align}

\section{SNR Statistics}
\label{SNR Statistics}
We commence by defining the closed-form CDFs of \eqref{snr1} and \eqref{snr2}, which represent the key statistics for the overall performance evaluation of the considered system.

\begin{lem}
The CDF of SNR for the $i^{\rm th}$ stream of Service 1 ($1\leq i \leq M_{1}$) is given by
\begin{align}
\nonumber
&F_{{\rm SNR}^{(1)}_{i}}(\gamma)=\\
&1-\sum^{N-M}_{k=0}\sum^{\rho(A)}_{v=1}\sum^{\tau_{v}(A)}_{j=1}\sum^{k}_{l=0}\frac{\Psi_{M}N^{k-l}_{0}\gamma^{k}\exp\left(-\frac{N_{0}\gamma}{a^{2}\hat{p}_{i}}\right)}{\left(\frac{\gamma}{a^{2}\hat{p}_{i}}+\frac{1}{\left\langle p_{v}-a^{2}\hat{p}_{v}\right\rangle}\right)^{j+l}},
\label{cdfsnr1}
\end{align}
where 
\begin{align}
\Psi_{M}\triangleq \binom{k}{l}\frac{\mathcal{X}_{v,j}(A)(j+l-1)!}{k!(a^{2}\hat{p}_{i})^{k}(j-1)!\left(\left\langle p_{v}-a^{2}\hat{p}_{v}\right\rangle\right)^{j}},
\label{psiM}
\end{align}
while $A\triangleq \diag\{p_{i}-a^{2}\hat{p}_{i}\}^{M}_{i=1}$; $\rho(A)$ denotes the number of distinct diagonal elements of $A$; $\left\langle p_{v}-a^{2}\hat{p}_{v}\right\rangle$ is a distinct diagonal element of $A$ in a descending order (e.g., $\left\langle p_{1}-a^{2}\hat{p}_{1}\right\rangle>\left\langle p_{2}-a^{2}\hat{p}_{2}\right\rangle>\cdots>\left\langle p_{v}-a^{2}\hat{p}_{v}\right\rangle$); $\tau_{v}(A)$ is the multiplicity of $\left\langle p_{v}-a^{2}\hat{p}_{v}\right\rangle$; and $\mathcal{X}_{v,j}(A)$ stands for the characteristic coefficient of $A$, explicitly provided in \cite[Eq. (129)]{j:ShinWin2008}.
\end{lem}

\begin{proof}
The proof is relegated in Appendix \ref{appa}.
\end{proof}

\begin{rem}
The CDF expression of \eqref{cdfsnr1} reflects to the general scenario of i.n.n.i.d. Rayleigh faded channels. Two extreme scenarios, where all the link distances are unequal and equal, can be modeled by i.n.i.d. and i.i.d. faded channels, correspondingly, which admit a more relaxed expression of \eqref{cdfsnr1}. Thus, for the i.n.i.d. case, $\rho(A)=M$ and $\tau_{v}(A)=1$. Similarly, for the i.i.d. case, $p_{i}-a^{2}\hat{p}_{i}\triangleq p-a^{2}\hat{p}\ \forall i$. Then, $\rho(A)=1$ and $\tau_{1}(A)=M$, while $\mathcal{X}_{v,j}(A)=1$ or $0$ when $j=M$ or $j<M$, respectively.   
\end{rem}

Next, we proceed to the derivation of the CDF of SNR for the remaining $M_{2}$ received streams of Service 2. To this end, the modeling of the colored noise vector $\mathbf{w}'$ is required. Following the same lines of reasoning as in \cite{j:Tandra,c:Kalamkar2013}, the noise uncertainty factor is modeled as $\beta=\hat{N}_{0}/N_{0}$. Let the upper bound on the noise uncertainty level be $L$ in dB, which is defined as
\begin{align}
L\triangleq {\rm sup}\{10 {\rm log}_{10}\beta\},
\label{Ldef}
\end{align}
while $\beta$ (in dB) is uniformly distributed in the range $\{-L,L\}$ \cite{j:Tandra}. Typically, $L\leq 2$dB for most practical applications \cite{c:Kalamkar2013}, whereas the PDF of $\beta$ is given by
\begin{align}
f_{\beta}(x)=\left\{
\begin{array}{c l}     
    \frac{5}{{\rm ln}(10)Lx}, & 10^{-\frac{L}{10}}<x<10^{\frac{L}{10}},\\
    & \\
    0, & {\rm otherwise}.
\end{array}\right.
\label{bpdf}
\end{align}

\begin{lem}
The CDF of SNR for the $i^{\rm th}$ stream of Service 2 ($1\leq i \leq M_{2}$) is given by
\begin{align}
\nonumber
&F_{{\rm SNR}^{(2)}_{i}}(\gamma)=1-\sum^{N-M_{2}}_{k=0}\sum^{\rho(A')}_{v=1}\sum^{\tau_{v}(A')}_{j=1}\sum^{k}_{l=0}\Psi_{M_{2}}\\
\nonumber
&\times \frac{5\gamma^{l}}{\left(\frac{\gamma}{a^{2}\hat{p}_{i}}+\frac{1}{\left\langle p_{v}-a^{2}\hat{p}_{v}\right\rangle}\right)^{j+l}L{\rm ln}(10)(a^{2}\hat{p}_{i})^{l-k}}\\
&\times \left[\Gamma\left(k-l,\frac{\hat{N}_{0}\gamma}{a^{2}\hat{p}_{i}10^{\frac{L}{10}}}\right)-\Gamma\left(k-l,\frac{\hat{N}_{0}\gamma}{a^{2}\hat{p}_{i}10^{-\frac{L}{10}}}\right)\right],
\label{cdfsnr2}
\end{align}
where $A'\triangleq \diag\{p_{i}-a^{2}\hat{p}_{i}\}^{M_{2}}_{i=1}$ and $\Psi_{M_{2}}$ is obtained from \eqref{psiM} by substituting $M$ with $M_{2}$ and $A$ with $A'$.
\end{lem}

\begin{proof}
The proof is provided in Appendix \ref{appb}.
\end{proof}

In the high SNR regime, the previously derived CDFs admit a more relaxed formulation. Actually, this case can be modeled by assuming that $p_{i}/N_{0}\rightarrow +\infty$. 

\begin{cor}
In the high SNR regime, $F_{{\rm SNR}^{(1)}_{i}}(\cdot)$ is sufficiently approximated from \eqref{cdfsnr1}, by neglecting the exponential function. Also, $F_{{\rm SNR}^{(2)}_{i}}(\cdot)$ is approximated as
\begin{align}
\nonumber
&F_{{\rm SNR}^{(2)}_{i}}(\gamma)\xrightarrow{\frac{p_{i}}{N_{0}}\rightarrow +\infty} 1-\sum^{N-M_{2}}_{k=0}\sum^{\rho(A')}_{v=1}\sum^{\tau_{v}(A')}_{j=1}\frac{\mathcal{X}_{v,j}(A')}{k!(a^{2}\hat{p}_{i})^{k}}\\
&\times \frac{(j+k-1)!5 \gamma^{k}\left[{\rm ln}\left(\frac{\hat{N}_{0}\gamma}{a^{2}\hat{p}_{i}10^{-\frac{L}{10}}}\right)-{\rm ln}\left(\frac{\hat{N}_{0}\gamma}{a^{2}\hat{p}_{i}10^{\frac{L}{10}}}\right)\right]}{(j-1)!\left(\left\langle p_{v}-a^{2}\hat{p}_{v}\right\rangle\right)^{j}\left(\frac{\gamma}{a^{2}\hat{p}_{i}}+\frac{1}{\left\langle p_{v}-a^{2}\hat{p}_{v}\right\rangle}\right)^{j+k}L{\rm ln}(10)},
\label{cdfsnr2approx}
\end{align} 
yielding CDF expressions including elementary-only functions.
\end{cor}

\begin{proof}
The approximation of $F_{{\rm SNR}^{(1)}_{i}}(\cdot)$ is straightforward since $\exp(0^{+})\rightarrow 1$. Regarding the approximation of $F_{{\rm SNR}^{(2)}_{i}}(\cdot)$, we know from \cite[Eqs. (8.352.7) and (8.359.1)]{tables} that $\Gamma(n,0^{+})\rightarrow (n-1)!$ for $n\in \mathbb{N}^{+}$, while $\Gamma(0,z)\rightarrow -{\rm ln}(z)-C$ for $z\rightarrow 0^{+}$, where $C$ denotes the Euler's constant \cite[Eq. (9.73)]{tables}. Hence, after some simple manipulations, \eqref{cdfsnr2approx} is obtained.
\end{proof}

\section{Performance Metrics}
\label{Performance Metrics}
Capitalizing on the previously derived results, a key performance metric is analytically evaluated, namely, the outage probability of the considered system. In addition, the rather interesting scenario of massive MIMO deployments is further studied and analyzed, where some useful outcomes are manifested; the effective coherence time and a necessary optimality condition, which preserves the given transmission quality for both service transmission modes.

In what follows, the following auxiliary notations will be used for ease of presentation
\begin{align}
\nonumber
F_{{\rm SNR}^{(1)}_{i}}(\gamma;T_{1})&\triangleq F_{{\rm SNR}^{(1)}_{i}}(\gamma),\\
F_{{\rm SNR}^{(2)}_{i}}(\gamma;T_{2})&\triangleq F_{{\rm SNR}^{(2)}_{i}}(\gamma),
\label{notationcdf}
\end{align}
where $T_{1}$ and $T_{2}$ denote the number of available transmit antennas/nodes, and are directly obtained by substituting $M$ with $T_{1}$ and $M_{2}$ with $T_{2}$ in \eqref{cdfsnr1} and \eqref{cdfsnr2}, respectively.

\subsection{Outage Probability of the General Case}
Outage probability denotes the probability that the system SNR falls bellow a specified threshold value $\gamma_{\rm th}$ and, hence, it is directly related to its corresponding CDF.

Due to the mode of the proposed operation, outage probability of the $M_{1}$ streams for Service 1 is defined as
\begin{align}
\nonumber
&P^{(1)}_{{\rm out},i}(\gamma_{\rm th})\triangleq {\rm Pr}\Bigg[\left({\rm SNR}^{(1)}_{i}\leq \gamma_{\rm th}\ {\rm and }\ {\rm SNR}^{(1)}_{\min}>\gamma_{\rm T}\right)\\
& {\rm or }\ \left({\rm SNR}^{(2)}_{i}\leq \gamma_{\rm th}\ {\rm and }\ {\rm SNR}^{(1)}_{\min}\leq \gamma_{\rm T}\right)\Bigg],\ 1\leq i\leq M_{1},
\label{outserv1}
\end{align}
where $\gamma_{\rm T}$ denotes the switching threshold, which is generally independent of $\gamma_{\rm th}$, while ${\rm SNR}^{(1)}_{\min}$ is the minimum SNR between the $M_{1}$ SNR values of Service 1. More specifically, \eqref{outserv1} indicates two exclusive events:

\begin{description}
\item[Event{\rm 1:}]\ The condition where all the received signals' SNR of Service 1 (i.e., $M_{1}$ signals) exceed the switching threshold $\gamma_{\rm T}$. In this case, these signals are detected/decoded in the first detection stage. Therefore, Event 1 can be modeled by the expression within the first parenthesis at the right-hand side (RHS) of \eqref{outserv1}.  
\item[Event{\rm 2:}]\ The condition where at least one of the received $M_{1}$ signals' SNR of Service 1 (i.e., the minimum SNR) falls below $\gamma_{\rm T}$. In this case, the switching process occurs amongst the two service groups. In fact, all the $M_{1}$ signals are detected/decoded at the second detection stage, after the detection and removal of $M_{2}$ signals. Hence, Event 2 can be modeled by the expression within the second parenthesis at the right-hand side (RHS) of \eqref{outserv1}.
\end{description}

\begin{prop}
Outage probability of the $i^{\rm th}$ received stream ($1\leq i\leq M_{1}$) for Service 1 is presented in a closed-form expression as
\begin{align}
\nonumber
P^{(1)}_{{\rm out},i}(\gamma_{\rm th})=&F_{{\rm SNR}^{(1)}_{i}}(\gamma_{\rm th};M)\bar{F}_{{\rm SNR}^{(1)}_{\min}}(\gamma_{\rm T};M)\\
&+F_{{\rm SNR}^{(2)}_{i}}(\gamma_{\rm th};M_{1})F_{{\rm SNR}^{(1)}_{\min}}(\gamma_{\rm T};M),
\label{outserv1cl}
\end{align}
with
\begin{align}
F_{{\rm SNR}^{(1)}_{\min}}(\gamma_{\rm T};M)\triangleq 1-\prod^{M_{1}}_{i=1}\left(1-F_{{\rm SNR}^{(1)}_{i}}(\gamma_{\rm T};M)\right).
\label{FSNRmin}
\end{align}
\end{prop}

\begin{proof}
Using the definition of \eqref{outserv1cl}, the derivations of \eqref{cdfsnr1} and \eqref{cdfsnr2}, while invoking the notation of \eqref{notationcdf}, the desired result in \eqref{outserv1cl} is directly obtained. Due to the fact that the received $M_{1}$ signals are mutually independent, the CDF of the minimum SNR follows the expression of \eqref{FSNRmin}.  
\end{proof}

Following a similar strategy, outage probability of the $M_{2}$ streams for Service 2 is defined as
\begin{align}
\nonumber
&P^{(2)}_{{\rm out},i}(\gamma_{\rm th})\triangleq {\rm Pr}\Bigg[\left({\rm SNR}^{(2)}_{i}\leq \gamma_{\rm th}\ {\rm and }\ {\rm SNR}^{(1)}_{\min}>\gamma_{\rm T}\right)\\
& {\rm or }\ \left({\rm SNR}^{(1)}_{i}\leq \gamma_{\rm th}\ {\rm and }\ {\rm SNR}^{(1)}_{\min}\leq \gamma_{\rm T}\right)\Bigg],1\leq i\leq M_{2}.
\label{outserv2}
\end{align}

Equation \eqref{outserv2} indicates the following two exclusive events:

\begin{description}
\item[Event{\rm 3:}]\ The condition where all the received signals' SNR of Service 1 (i.e., $M_{1}$ signals) exceed the switching threshold $\gamma_{\rm T}$. In this case, the $M_{2}$ signals of Service 2 are detected/decoded at the second detection stage. Therefore, Event 3 can be modeled by the expression within the first parenthesis at the RHS of \eqref{outserv2}.  
\item[Event{\rm 4:}]\ The condition where at least one of the received $M_{1}$ signals' SNR of Service 1 (i.e., the minimum SNR) falls below $\gamma_{\rm T}$. In this case, all the $M_{2}$ signals are detected/decoded in the first detection stage. Hence, Event 4 can be modeled by the expression within the second parenthesis at the right-hand side (RHS) of \eqref{outserv1}.
\end{description}

\begin{prop}
Outage probability of the $i^{\rm th}$ received stream ($1\leq i\leq M_{2}$) for Service 2 is given in a closed-form expression by
\begin{align}
\nonumber
P^{(2)}_{{\rm out},i}(\gamma_{\rm th})=&F_{{\rm SNR}^{(2)}_{i}}(\gamma_{\rm th};M_{2})\bar{F}_{{\rm SNR}^{(1)}_{\min}}(\gamma_{\rm T};M)\\
&+F_{{\rm SNR}^{(1)}_{i}}(\gamma_{\rm th};M)F_{{\rm SNR}^{(1)}_{\min}}(\gamma_{\rm T};M).
\label{outserv2cl}
\end{align}
\end{prop}

\begin{proof}
The proof follows the same lines of reasoning as for the derivation of \eqref{outserv1cl}.  
\end{proof}

A total outage event (of each service group) in spatial multiplexed systems is declared when at least one of the multiple simultaneously transmitted signals is in outage. As such, the total system outage probability for each service group is defined as
\begin{align}
P^{(j)}_{{\rm out},\rm Total}(\gamma_{\rm th})\triangleq 1-\prod^{M_{j}}_{i=1}\left[1-P^{(j)}_{{\rm out},i}(\gamma_{\rm th})\right], \ \ j=\{1,2\}.
\end{align}

\subsection{Massive MIMO Deployment}
Next, we study the scenario of a dense multiuser communication system, whereas the receiver is equipped with a vast number (in the order of tens or a few hundreds) of antenna elements. From the information theoretic standpoint, such a scenario can be modeled by assuming that $\{N\ {\rm and/or}\ M\} \rightarrow +\infty$.

In this asymptotic regime, the SNR expressions in \eqref{snr1} and\eqref{snr2} admit a more relaxed formulation. Specifically, using \cite[Eq. (37)]{j:NgoMatthaiou2013}, we get
\begin{align}
{\rm SNR}^{(1)}_{i}\xrightarrow{\{N,M\}\rightarrow +\infty} \frac{(N-M)a^{2}\hat{p}_{i}}{\sum^{M}_{j=1}(p_{j}-a^{2}\hat{p}_{j})},\ 1\leq i\leq M_{1},
\label{snr1asy}
\end{align}
and 
\begin{align}
{\rm SNR}^{(2)}_{i}\xrightarrow{\{N,M_{2}\}\rightarrow +\infty} \frac{(N-M_{2})a^{2}\hat{p}_{i}}{\sum^{M_{2}}_{j=1}(p_{j}-a^{2}\hat{p}_{j})},\ 1\leq i\leq M_{2},
\label{snr2asy}
\end{align}
with fixed ratios $N/M<+\infty$ and $N/M_{2}<+\infty$.

It is clear from \eqref{snr1asy} and \eqref{snr2asy} that the impact of noise is negligible in massive MIMO deployments. Only in the limited case when $N\gg M$ and $N\gg M_{2}$ (i.e., for an asymptotically high $N$ and fixed $M$), the above deterministic SNR equivalents grow up to infinity without any bound. 

\subsubsection{\underline{Optimal Value of $M_{2}$}}
The overall system performance is directly related to the number of users/transmit antennas of each service group, which are involved in a given time frame. Given that the priority is to preserve the transmission quality of Service 1, $M_{1}$ streams can reach an arbitrary number up to the upper bound $M_{1}\leq M$. In the latter extreme case, $M_{2}=0$ (since $M=M_{1}+M_{2}$), denoting that Service 1 fully occupies the system resources, resulting to no available space for Service 2.

On the other hand, arbitrarily setting the range of $M_{2}$ as $[M_{1}+1,N]$ (such that $M\leq N$) may not always be preferable and/or correspond to the optimal solution. This is because adding more transmit antennas, the available degrees-of-freedom at the received are being reduced, reflecting on a performance degradation for Service 1. Keeping in mind that the objective in massive MIMO deployments is to maintain the $M_{1}$ signals of Service 1 in the first detection stage, it is desirable to determine the optimal $M^{\star}_{2}$, such that $M^{\star}=M_{1}+M^{\star}_{2}$.

This optimization problem can be formulated as
\begin{align}
\nonumber
\max_{M_{2}>0}&\ \sum^{M_{2}}_{i=1}{\rm log}_{2}\left(1+{\rm SNR}^{(2)}_{i}\right)\\
{\rm s.t.}&\ \ {\rm SNR}^{(1)}_{\min}\geq \gamma_{\rm th},\ \ \ {\rm SNR}^{(1)}_{\min}=\min\{{\rm SNR}^{(1)}_{j}\}^{M_{1}}_{j=1}.
\label{optpro}
\end{align}
Further, in order to obtain a tractable solution and corresponding engineering insights, i.i.d. channel fading conditions regarding the $M_{2}$ received signals of Service 2 are assumed. Then, using \eqref{snr2asy} and after some simple manipulations, the above optimization problem can be reformulated as
\begin{align}
\nonumber
\max_{M_{2}>0}&\ M_{2}\:{\rm log}_{2}\left(1+\frac{\left(\frac{N}{M_{2}}-1\right)a^{2}\hat{p}}{p-a^{2}\hat{p}}\right)\\
{\rm s.t.}&\ \ \frac{\gamma_{\rm th}}{\left(\frac{\left(N-M_{1}-M_{2}\right)a^{2}\hat{p}}{(M_{1}+M_{2})(p-a^{2}\hat{p})}\right)}-1\leq 0.
\label{optpro1}
\end{align}

The optimization problem in \eqref{optpro1} is convex since the objective function and its constraint are strictly concave and strictly convex, respectively, since 
\begin{align}
\nonumber
&\frac{\partial^{2}}{\partial^{2}M_{2}}\left(M_{2}\:{\rm log}_{2}\left(1+\frac{\left(\frac{N}{M_{2}}-1\right)a^{2}\hat{p}}{p-a^{2}\hat{p}}\right)\right)=\\
&-\frac{(a^{2}\hat{p}N)^{2}}{M^{3}_{2}(p-a^{2}\hat{p})^{2}\left(1+\frac{\left(\frac{N}{M_{2}}-1\right)a^{2}\hat{p}}{p-a^{2}\hat{p}}\right)^{2}{\rm ln}(2)}<0,
\label{concave}
\end{align}
and
\begin{align}
\nonumber
&\frac{\partial^{2}}{\partial^{2}M_{2}}\left(\frac{\gamma_{\rm th}}{\left(\frac{\left(N-M_{1}-M_{2}\right)a^{2}\hat{p}}{(M_{1}+M_{2})(p-a^{2}\hat{p})}\right)}-1\right)=\\
&\frac{2(p-a^{2}\hat{p})\gamma_{\rm th}}{a^{2}\hat{p}(N-M_{1}-M_{2})^{2}}+\frac{2(M_{1}+M_{2})(p-a^{2}\hat{p})\gamma_{\rm th}}{a^{2}\hat{p}(N-M_{1}-M_{2})^{3}}>0,
\label{convex}
\end{align}
which implies that the Karush-Kuhn-Tucker conditions are both \emph{necessary and sufficient} for the optimal solution, whereas a unique value exists \cite{b:ConvexOptimization}. Motivated by the convexity of the problem, we introduce the following Lagrangian multiplier, termed $\lambda$, into the equation:
\begin{align}
\nonumber
\mathcal{L}\triangleq &M_{2}\:{\rm log}_{2}\left(1+\frac{\left(\frac{N}{M_{2}}-1\right)a^{2}\hat{p}}{p-a^{2}\hat{p}}\right)\\
&-\lambda \left(\frac{\gamma_{\rm th}}{\left(\frac{\left(N-M_{1}-M_{2}\right)a^{2}\hat{p}}{(M_{1}+M_{2})(p-a^{2}\hat{p})}\right)}-1\right), \ \lambda\geq 0,
\label{Lagrange}
\end{align}
and we set
\begin{align}
\frac{\partial \mathcal{L}}{\partial M_{2}}=0.
\end{align}
Solving the latter equality yields
\begin{align}
\nonumber
\lambda=&\frac{a^{2}\hat{p}(N-M_{1}-M^{\star}_{2})^{2}}{{\rm ln}(2)\gamma_{\rm th}(p-a^{2}\hat{p})\left(a^{2}\hat{p}(N-M^{\star}_{2})+M^{\star}_{2}(p-a^{2}\hat{p})\right)}\\
\nonumber
&\times \Bigg[{\rm ln}\left(1+\frac{\left(\frac{N}{M^{\star}_{2}}-1\right)a^{2}\hat{p}}{p-a^{2}\hat{p}}\right)\\
&\times \left(a^{2}\hat{p}(N-M^{\star}_{2})+M^{\star}_{2}(p-a^{2}\hat{p})\right)-a^{2}\hat{p}N\Bigg].
\label{lambda}
\end{align}
It is should be stated at this point that finding $M^{\star}_{2}$, which is the unique optimal value of \eqref{optpro1}, does not represent our primary goal, due the mode of the proposed operation. In fact, it is much more insightful to provide an optimality condition whereon the problem in \eqref{optpro1} is solved. Since $\lambda\geq 0$, the aforementioned optimality condition reads from \eqref{lambda} as
\begin{align}
\nonumber
&{\rm ln}\left(1+\frac{\left(\frac{N}{M^{\star}_{2}}-1\right)a^{2}\hat{p}}{p-a^{2}\hat{p}}\right)\\
&\times \left(a^{2}\hat{p}(N-M^{\star}_{2})+M^{\star}_{2}(p-a^{2}\hat{p})\right) \geq a^{2}\hat{p}N.
\label{optcond}
\end{align}

\begin{rem}
It is remarkable that the condition of \eqref{optcond} is independent of the number of transmitted streams (i.e., $M_{1}$) for Service 1 and the outage threshold $\gamma_{\rm th}$ (i.e., the minimum target on the achievable data rate for Service~1). On the other hand, it is directly related to the number of receive antennas and the channel fading statistics.
\end{rem}

Notice that \eqref{optcond} is a linear and rather simple expression since it includes elementary-only functions. Thus, it can be computed very rapidly and efficiently at the reception stage prior to each frame transmission. Actually, given $\{M_{1},M_{2},N\}$ and the corresponding channel statistics, derived from the training phase, the receiver applies \eqref{optcond}. If the condition is satisfied, then $M_{2}$ streams are allowed for transmission (to be detected at the second stage) without causing any problem to the transmission quality of Service 1. In the case when \eqref{optcond} is not satisfied, the receiver applies the former procedure for $M_{2}-1$. This is repeated until the optimality condition is satisfied. Finally, the value of $M^{\star}_{2}$ is dictated and used for the next frame transmission interval. 

For completeness of exposition, the proposed iterative approach is formalized in Algorithm~2.

\begin{algorithm}[t]
	\caption{Define the Optimal Value of $M_{2}$ (i.e., $M^{\star}_{2}$)}
	\begin{algorithmic}[1]
		 \INPUT{$M_{1}$, $M_{2}$, $N$, and the corresponding channel statistics}
		 \OUTPUT{$M^{\star}_{2}$}
		  \WHILE{$M_{2}>0$}
			\STATE{Compute the optimality condition as per \eqref{optcond}}
			\IF{\eqref{optcond} holds true}
				\STATE $M^{\star}_{2}=M_{2}$;
				\STATE End of the algorithm;
			\ELSE $\:\:M_{2}=M_{2}-1$
				\STATE Go to Step 2;
			\ENDIF
		\ENDWHILE 
	\end{algorithmic}
\end{algorithm}

\subsubsection{\underline{Coherence Time}}
The effective coherence time of the system represents quite an important metric, which determines the overall system performance. It denotes the amount of time where the channel fading conditions remain unchanged; thus, it determines the efficiency of the training phase prior to data transmission/reception. For the $n^{\rm th}$ time instant, SNR decays with $n$. Hence, the usable transmission frame duration should always satisfy ${\rm SNR}^{(1)}_{\min}[n]\geq \gamma_{\rm th}$, such that ${\rm SNR}^{(1)}_{\min}[T_{\max}]=\gamma_{\rm th}$, where $T_{\max}$ denotes the coherence time in terms of the maximum number of consecutive transmitted symbols.\footnote{Recall that the priority is to preserve the transmission quality of Service 1, while optimizing the detection/decoding processing delay at the same time. It turns out that such criteria are met by allocating Service 1 to the first detection stage (in massive MIMO).}

To obtain mathematical tractability, the autoregressive model of order one was used in the previous analysis. Yet, according to \eqref{jakesorderm} and \eqref{errorderm}, we can easily extend it to the general scenario of order $T_{\max}$. Doing so, \eqref{snr1asy} becomes
\begin{align}
{\rm SNR}^{(1)}_{\min}[T_{\max}]\xrightarrow{\{N,M\}\rightarrow +\infty} \frac{(N-M)a^{2 T_{\max}}\hat{p}_{\min}}{\sum^{M}_{i=1}(p_{i}-a^{2 T_{\max}}\hat{p}_{i})}=\gamma_{\rm th}.
\label{snr1asyy}
\end{align}
Equation \eqref{snr1asyy} can be further simplified since, from \eqref{esterror}, it holds that $\hat{p}_{i}=p_{i}$ $\forall i$, when $M\rightarrow +\infty$. 

\begin{cor}
The effective channel coherence time in terms of $T_{\max}$, when $N$ and/or $M\rightarrow +\infty$, yields as
\begin{align}
\nonumber
&\frac{(N-M)a^{2 T_{\max}}\hat{p}_{\min}}{(1-a^{2 T_{\max}})\left(\sum^{M}_{i=1}\hat{p}_{i}\right)}=\gamma_{\rm th}\\
\Longleftrightarrow &T_{\max}=\left\lfloor \frac{{\rm ln}\left(\frac{\gamma_{\rm th}\left(\sum^{M}_{i=1}\hat{p}_{i}\right)}{\gamma_{\rm th}\left(\sum^{M}_{i=1}\hat{p}_{i}\right)+\hat{p}_{\min}(N-M)}\right)}{2{\rm ln}(a)}\right\rfloor.
\label{snr1asyyyy}
\end{align}
In the simplified scenario when all the involved link distances are identical, a simple closed-form expression of $T_{\max}$ is given by
\begin{align}
T_{\max}=\left\lfloor \frac{{\rm ln}\left(\frac{M \gamma_{\rm th}}{(N-M+M\gamma_{\rm th})}\right)}{2{\rm ln}(a)}\right\rfloor.
\label{tmax}
\end{align}
\end{cor}

\begin{rem}
By closely observing \eqref{tmax}, the effective channel coherence time increases with $N$ and $a$ (i.e., a slower channel aging effect), while it decreases with $M$; all in a logarithmic scale. 
\label{rem3} 
\end{rem}

\section{Numerical Results}
\label{Numerical Results}
In this section, numerical results are presented and cross-compared with Monte-Carlo simulations to assess our theoretical findings. Line-curves and circle-marks denote the analytical and simulation results, respectively. Also, solid and dashed lines denote the analytical results regarding the primary (Service~1) and secondary nodes (Service~2), correspondingly. There is a perfect match between the former and latter results and, therefore, the accuracy of the presented analysis is verified. Henceforth, without loss of generality, we set $\omega_{i}=4\:\forall i$, corresponding to a classical macro-cell urban environment, while the upper bound on residual noise uncertainty is $L=2$dB. Also, unless otherwise specified, $\gamma_{\rm th}=\gamma_{\rm T}=1$, i.e., reflecting a minimum target data rate of $1$bps/Hz for Service~1, and $p_{\rm T}/N_{0}=10$dB.   

In Fig.~\ref{fig2}, the system outage probability is presented vs. various link distances, assuming i.i.d. channel fading conditions for all the involved nodes. For the far-distant transmission regions, the primary system (Service~1) always experiences a better outage performance in comparison to the secondary system (Service~2), as it should be. This occurs because, in these regions, the switching scheme allocates the primary signals in the second detection stage (on average) due to their relatively low received SNR. On the other hand, the outage performance is increased for both service groups as the link distances are reduced. In particular, the secondary service outperforms the primary one, in these regions, since the corresponding signals are being detected in the second stage, due to the higher received SNR conditions. Importantly, the presence of a more intense channel aging effect (e.g., higher user mobility) results to a dramatic performance degradation of both services. A typical outage requirement for most practical applications (e.g., $\leq 10^{-2}$) is satisfied when $d\leq 150$m. However, such a restricted coverage area denotes a classical region of operation for most heterogeneous modern network systems, especially in dense urban environments.
\begin{figure}[!t]
\centering
\includegraphics[trim=1.5cm 0.5cm 1.5cm 0.0cm, clip=true,totalheight=
0.5\textheight]{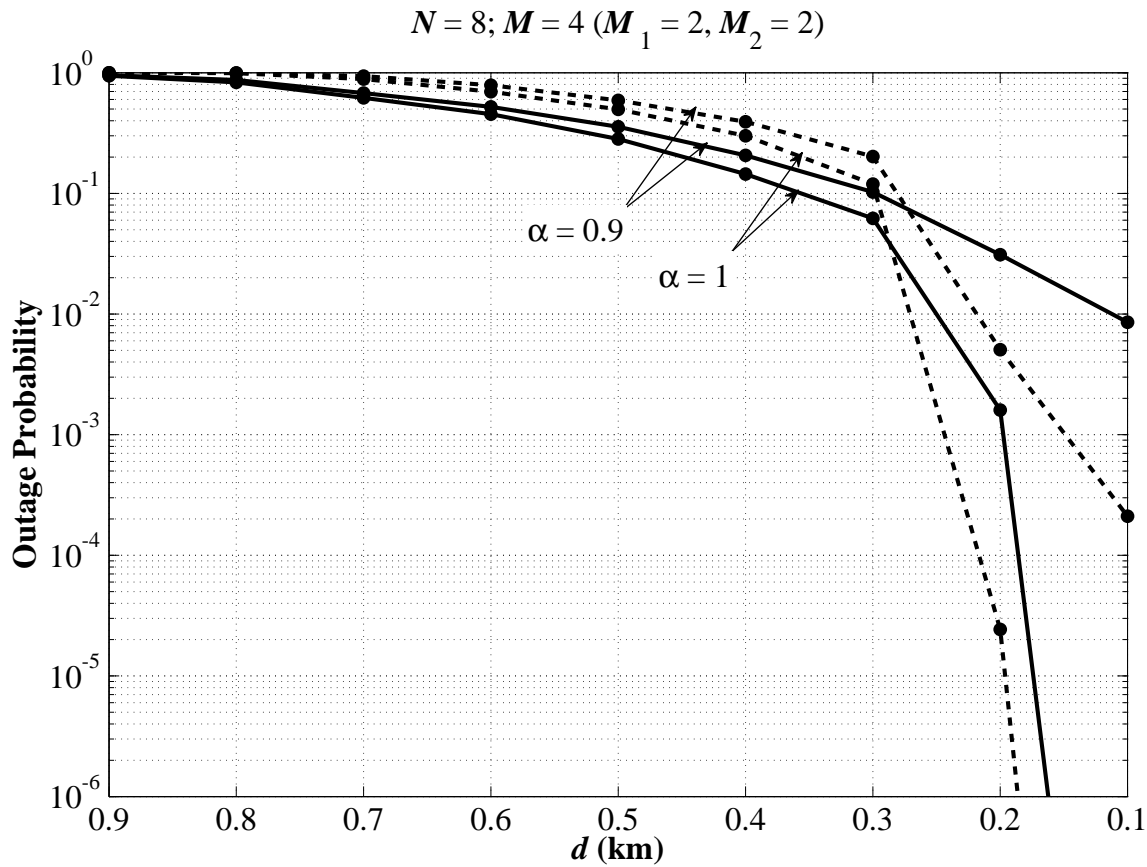}
\caption{Outage probability vs. various link distances for i.i.d. channel fading conditions.}
\label{fig2}
\end{figure}

Figure~\ref{fig3} illustrates the scenario of a high-rate system in a closer vicinity under i.n.n.i.d. channel fading conditions. In fact, this scenario can be visualized by considering a single primary and a single secondary user, each equipped with two transmit antennas and arbitrary distances with respect to the receiver. The system outage performance is presented vs. different number of transmit antennas per service. Recall that $M_{2}=M-M_{1}$, while $M=4$ is fixed. The presence of a more intense channel aging effect emphatically impacts on the outage performance of both services, as expected. It is also noteworthy that a higher $\alpha$ value (i.e., a slower channel aging) reflects on a more accurate instantaneous CSI, thus a more effective detection. As $M_{1}$ increases, the possibility of at least one of the $M_{1}$ SNRs for Service~1 falling below $\gamma_{\rm T}$ is higher. Doing so, the probability that are being detected in the second stage also increases (which in turn results to a better outage performance). However, as the $\alpha$ value is reduced (i.e., the erroneous CSI increases) the estimated SNR of all the received signals experiences a more dramatic corruption, yielding to bad switching decisions more often. In the latter case (e.g., when $\alpha=0.6$), a higher $M_{1}$ value reflects on a reduced outage performance. Thereby, it is obvious that the channel aging effect plays quite a critical role in the overall system performance.
\begin{figure}[!t]
\centering
\includegraphics[trim=1.5cm 0.5cm 1.5cm 0.0cm, clip=true,totalheight=
0.5\textheight]{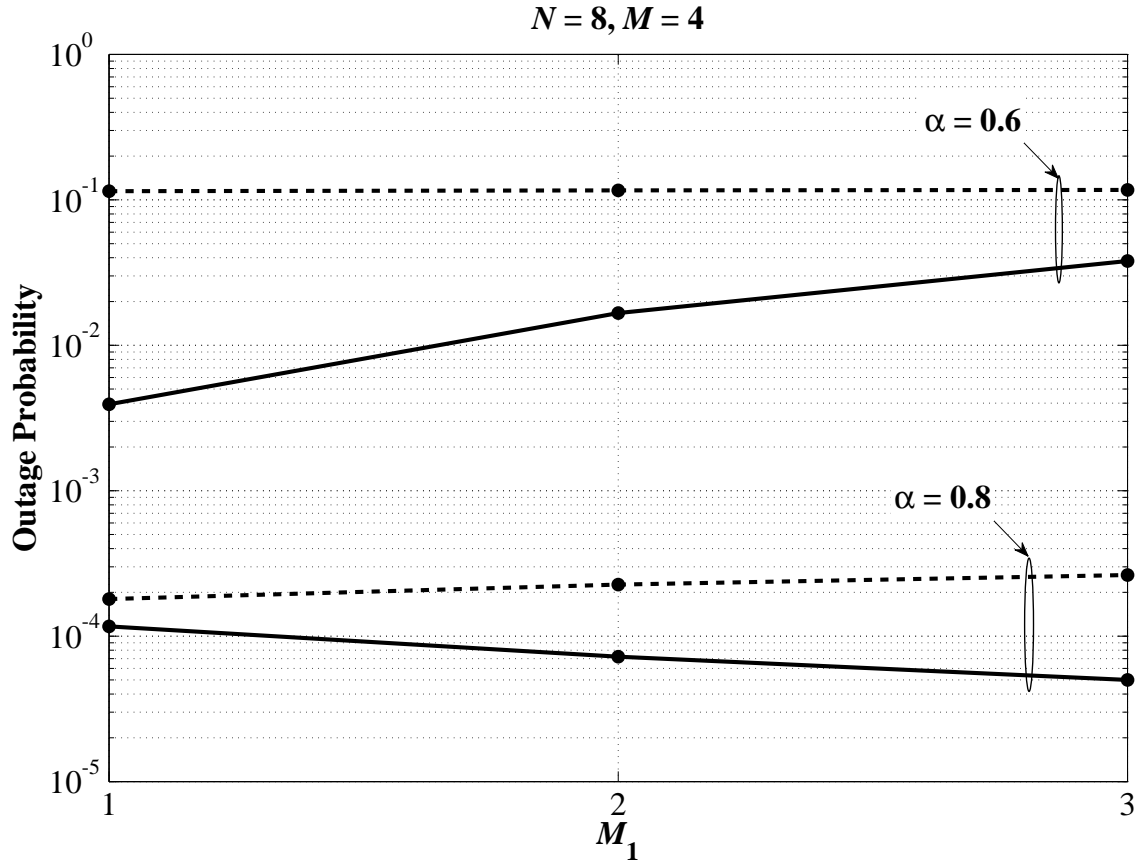}
\caption{Outage probability vs. various $M_{1}$ values for i.n.n.i.d. channel fading conditions. Particularly, the scenario of co-located antennas of a single primary and a single secondary node is illustrated. The link distances from the primary and secondary nodes to the receiver are respectively given by $\{d_{i}\}^{M_{1}}_{i=1}=25$m. and $\{d_{j}\}^{M_{2}}_{j=1}=35$m. Also, $\gamma_{\rm T}=7$ is assumed reflecting to a minimum target data rate of 3bps/Hz for the primary service.}
\label{fig3}
\end{figure}

In Fig.~\ref{fig4}, the ideal scenario of a perfect and non-delayed CSI is assumed (i.e., when $\alpha=1$ and $\hat{p}_{i}=p_{i}$ $\forall i$). The system outage performance vs. different input SNR values is illustrated. The `crossing' of the illustrated outage curves between the two services is due to the proposed switching scheme at the detector, as previously analyzed and described. Clearly, the outage performance improves for higher input SNR values. Yet, such an ideal behavior is not realistic in practice, whereas it can only serve as a performance benchmark.\footnote{For realistic imperfect and/or delayed CSI conditions, outage probability reaches to a certain non-zero performance floor in high SNR regions \cite{j:MiridakisTsiftsisRowell2017,j:MiridakisTsiftsisAlexandropoulos2017}.}  
\begin{figure}[!t]
\centering
\includegraphics[trim=1.5cm 0.0cm 1.5cm 0.0cm, clip=true,totalheight=
0.5\textheight]{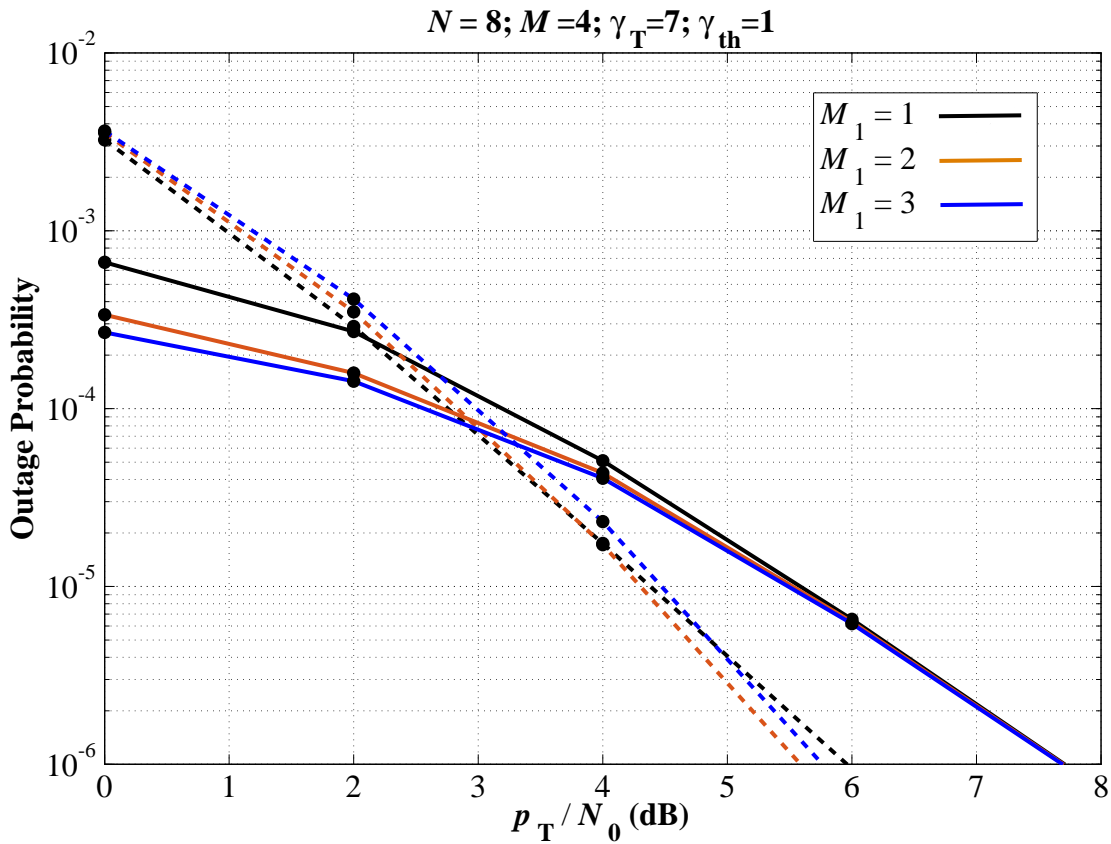}
\caption{Outage probability vs. various $p_{\rm T}/N_{0}$ values for i.n.n.i.d. channel fading conditions. Particularly, the scenario of co-located antennas of a single primary and a single secondary node is illustrated as in Fig. \ref{fig4}. The link distances from the primary and secondary nodes to the receiver are respectively given by $\{d_{i}\}^{M_{1}}_{i=1}=0.95$km. and $\{d_{j}\}^{M_{2}}_{j=1}=0.75$km. Also, an ideal CSI is assumed (i.e., $\alpha=1$ and $\hat{p}_{i}=p_{i}$ $\forall i$).}
\label{fig4}
\end{figure}

Figure~\ref{fig5} compares the outage performance of the primary service with or without the presence of secondary users. The selected scenario assumes a low antenna range at the receiver (i.e., $N=4$) and i.i.d. channel fading conditions in order to straightforwardly extract useful outcomes. Also, a high-rate system (a minimum target on data rate is set as $4$bps/Hz) is assumed, operating in a low SNR region. Although the system outage performance for the primary service is reduced when adding one or two secondary transmitting signals, it still remains in relatively low levels. In fact, outage probability is quite small for restricted coverage heterogeneous networks (e.g., when $d\leq 30$m). It should be stated that a much better outage performance has been reached for higher SNR regions and/or when more receive antennas are included.
\begin{figure}[!t]
\centering
\includegraphics[trim=1.5cm 0.0cm 1.5cm 0.0cm, clip=true,totalheight=
0.5\textheight]{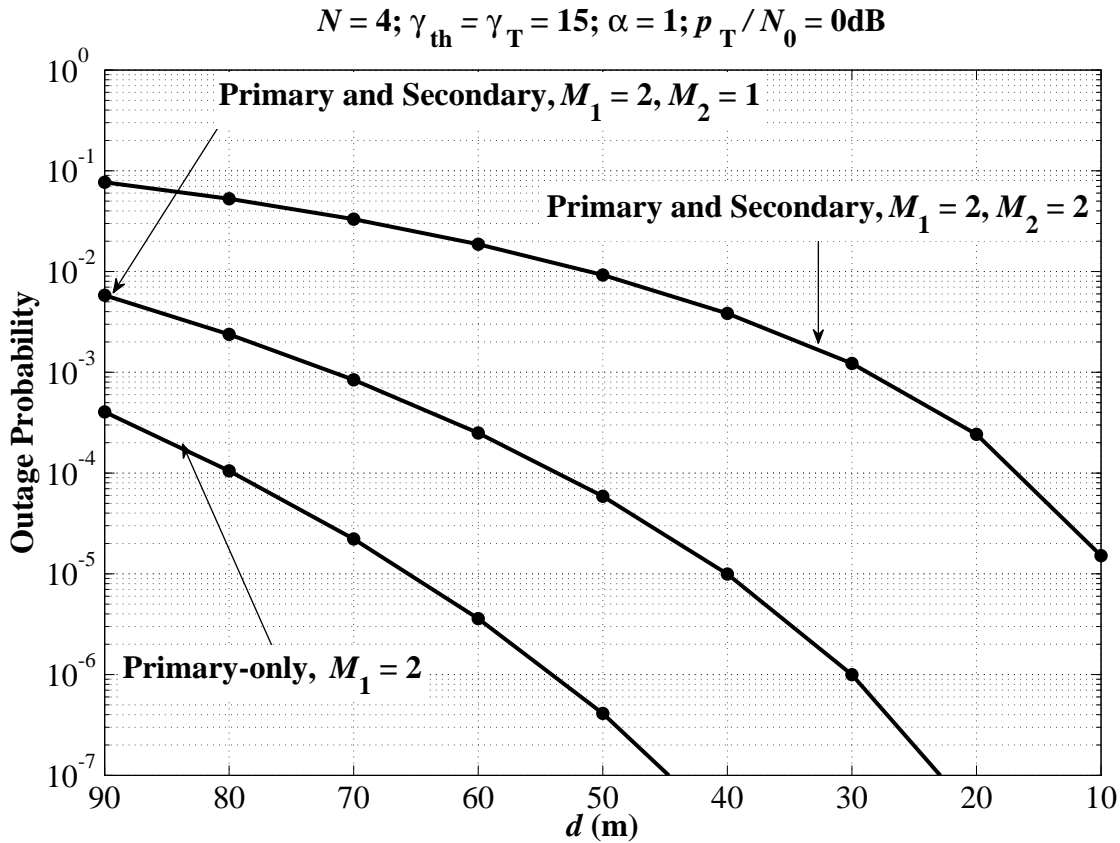}
\caption{Outage probability of the primary service (Service~1) vs. various link distances for i.i.d. channel fading conditions, where $d=100$m.}
\label{fig5}
\end{figure}

We proceed on the numerical results of a massive MIMO system, when $\{N,M\}\rightarrow +\infty$. Table~\ref{table} provides the maximum admissible number of simultaneously transmitted secondary signals (by utilizing Algorithm~2) for a given $\{N,\alpha\}$ set. Obviously, $M^{\star}_{2}$ is not always equal to $N-M_{1}$, especially when $\alpha$ is reduced indicating an increased channel aging effect. Hence, even in massive MIMO systems with favorable propagation conditions (i.e., when fast fading is averaged out), finding the exact $M^{\star}_{2}\leq N-M_{1}$ represents a critical issue.   

\begin{table}[t]
\caption {Optimal Value of Secondary Transmit Antennas} 
\begin{center}
\begin{tabular}{l| l l l l l l}\hline
$\boldsymbol{N=}$ & $\mathbf{16}$ &$\mathbf{32}$ & $\mathbf{64}$& $\mathbf{128}$ &$\mathbf{256}$ & $\mathbf{512}$\\\hline
 $\:\boldsymbol{\alpha=0.9999}$ &$6$ &$22$ &$54$ &$116$&$232$&$464$ \\
 $\:\boldsymbol{\alpha=0.8}$ &$6$ &$13$ &$27$ &$54$&$108$&$217$ \\
 $\:\boldsymbol{\alpha=0.6}$ &$5$ &$10$ &$20$ &$40$&$80$&$160$ \\\hline
\end{tabular}
\end{center}
*We consider that $M_{1}=10$; hence, the upper bound of $M_{2}$ is $N-10$.
\label{table}
\end{table}

Finally, the effective channel coherence time is presented in Fig.~\ref{fig6} for various system setups. Recall that $T_{\max}$ denotes the maximum number of transmit time symbols, where CSI remains unchanged between two consecutive training phases. The impact of channel aging dramatically affects the channel coherence time and, thereby, the actual system throughput. Additionally, it is obvious that all the provided statements in Remark~\ref{rem3} are verified.
\begin{figure}[!t]
\centering
\includegraphics[trim=1.5cm 0.0cm 1.5cm 0.0cm, clip=true,totalheight=
0.5\textheight]{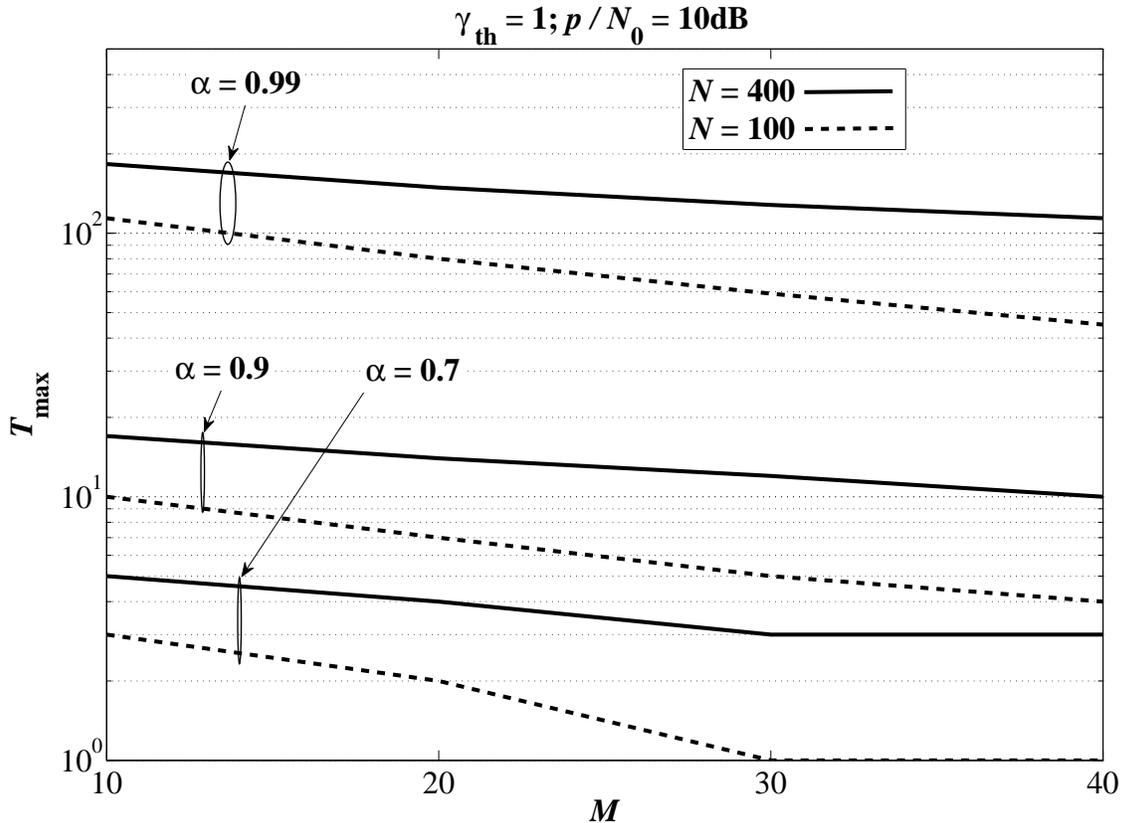}
\caption{Effective channel coherence time in terms of maximum number of consecutive transmit symbols $T_{\max}$ vs. different number of transmitters for a massive MIMO antenna array.}
\label{fig6}
\end{figure}

\section{Conclusion}
\label{Conclusion}
A new detection scheme for multiuser MIMO systems with different priorities was thoroughly studied and presented. A particular emphasis was given in CR-enabled communication, where the primary and secondary users can be grouped together in two distinct service groups, correspondingly. A two-stage detection scheme was proposed, which implements a certain switching between the received signals of each service group. The main objective is the performance improvement of both services and the overall multiuser diversity enhancement. Further, the scenario of massive MIMO systems was studied, where the maximum number of admissible secondary users is obtained via a necessary and sufficient optimality condition. Finally, some useful engineering insights were manifested, such as the impact of imperfect/delayed CSI and the size of MIMO antenna array onto the system performance.

\appendix
\subsection{Derivation of (\ref{cdfsnr1})}
\label{appa}
\numberwithin{equation}{subsection}
\setcounter{equation}{0}
It holds from \eqref{snr1} that 
\begin{align}
{\rm SNR}^{(1)}_{i}=\frac{\left\|\left[\left(\hat{\mathbf{H}}\hat{\mathbf{P}}^{\frac{1}{2}}\right)^{\dagger}\right]_{i}\right\|^{-2}}{\frac{\left\|\left[\left(\hat{\mathbf{H}}\hat{\mathbf{P}}^{\frac{1}{2}}\right)^{\dagger}\right]_{i}\mathbf{E}\right\|^{2}}{\left\|\left[\left(\hat{\mathbf{H}}\hat{\mathbf{P}}^{\frac{1}{2}}\right)^{\dagger}\right]_{i}\right\|^{2}}+N_{0}}\overset{\text{d}}=\frac{X_{i}}{Y_{i}+N_{0}},\ \ i\in [1,M_{1}],
\label{snr1distr}
\end{align}
where the second equality (in distribution) directly arises from \cite[Proposition 1]{j:NgoMatthaiou2013}, $X_{i}$ follows a chi-squared distribution with $2(N-M+1)$ degrees-of-freedom, and $Y_{i}$ denotes the sum of i.n.n.i.d. exponential RVs. More specifically, we have that
\begin{align}
F_{X_{i}}(x)=1-\exp\left(-\frac{x}{a^{2}\hat{p}_{i}}\right)\sum^{N-M}_{k=0}\frac{x^{k}}{k!(a^{2}\hat{p}_{i})^{k}},
\label{XiCDF}
\end{align}
and
\begin{align}
f_{Y_{i}}(y)=\sum^{\rho(A)}_{v=1}\sum^{\tau_{v}(A)}_{j=1}\frac{\mathcal{X}_{v,j}(A)y^{j-1}\exp\left(-\frac{y}{\left\langle p_{v}-a^{2}\hat{p}_{v}\right\rangle}\right)}{(j-1)!\left(\left\langle p_{v}-a^{2}\hat{p}_{v}\right\rangle\right)^{j}}.
\label{YiPDF}
\end{align}
Hence, according to \eqref{snr1distr}, we get
\begin{align}
F_{{\rm SNR}^{(1)}_{i}}(\gamma)=\int^{\infty}_{0}F_{X_{i}}(\gamma(y+N_{0}))f_{Y_{i}}(y)dy,\ \ i\in [1,M_{1}].
\label{snr1distrCDF}
\end{align}
Then, inserting \eqref{XiCDF} and \eqref{YiPDF} into \eqref{snr1distrCDF}, utilizing the binomial expansion, and after performing some straightforward manipulations, we arrive at the desired result of \eqref{cdfsnr1}.

\subsection{Derivation of (\ref{cdfsnr2})}
\label{appb}
\numberwithin{equation}{subsection}
\setcounter{equation}{0}
Following a similar methodology as in Appendix \ref{appa}, while conditioning on $\beta$, the CDF of SNR for the $i^{\rm th}$ stream of Service 2 is derived as
\begin{align}
\nonumber
&F_{{\rm SNR}^{(2)}_{i}|\beta}(\gamma|\beta)=\\
&1-\sum^{N-M_{2}}_{k=0}\sum^{\rho(A')}_{v=1}\sum^{\tau_{v}(A')}_{j=1}\sum^{k}_{l=0}\frac{\Psi_{M_{2}}\left(\frac{\hat{N}_{0}}{\beta}\right)^{k-l}\gamma^{k}\exp\left(-\frac{\hat{N}_{0}\gamma}{\beta a^{2}\hat{p}_{i}}\right)}{\left(\frac{\gamma}{a^{2}\hat{p}_{i}}+\frac{1}{\left\langle p_{v}-a^{2}\hat{p}_{v}\right\rangle}\right)^{j+l}}.
\label{cdfsnr2cond}
\end{align}
Thereby, the corresponding unconditional CDF yields as
\begin{align}
F_{{\rm SNR}^{(2)}_{i}}(\gamma)&=\int^{10^{\frac{L}{10}}}_{10^{-\frac{L}{10}}}F_{{\rm SNR}^{(2)}_{i}|\beta}(\gamma|z)f_{\beta}(z)dz.
\label{cdfsnr2ucond}
\end{align}
Inserting \eqref{cdfsnr2cond} and \eqref{bpdf} into \eqref{cdfsnr2ucond}, while using \cite[Eq. (3.351.2)]{tables}, the desired result of \eqref{cdfsnr2} is extracted.

\ifCLASSOPTIONcaptionsoff
  \newpage
\fi

\bibliographystyle{IEEEtran}
\bibliography{IEEEabrv,References}

\end{document}